\documentclass[journal,onecolumn,12pt]{IEEEtran}
\usepackage[numbers,sort&compress]{natbib}
\usepackage{pslatex}
\usepackage{amsfonts,color,morefloats}
\usepackage{amssymb,amsmath,latexsym,amsthm}
\usepackage{hyperref}
\usepackage{color}
\usepackage{makecell}
\hypersetup{hidelinks}
\allowdisplaybreaks

\newcommand{\gf}{{\mathbb{F}}}

\newtheorem{theorem}{Theorem}
\newtheorem{lemma}{Lemma}

\newtheorem{definition}{Definition}
\newtheorem{example}{Example}

\newtheorem{remark}{Remark}

\begin{document}
\title{A note on the differential spectrum of a class of locally APN functions
	}
	\author{Haode Yan\IEEEauthorrefmark{1}, Ketong Ren\IEEEauthorrefmark{1}\\
	\IEEEauthorblockA{\IEEEauthorrefmark{1}School of Mathematics, Southwest Jiaotong University, Chengdu, China.\\ E-mail:  \href{mailto: hdyan@swjtu.edu.cn}{hdyan}@swjtu.edu.cn(corresponding author), \href{mailto: rkt@my.swjtu.edu.cn}{rkt}@my.swjtu.edu.cn,}\\
	}
	\maketitle
	\begin{abstract}
		Let $\gf_{p^n}$ denote the finite field containing $p^n$ elements, where $n$ is a positive integer and $p$ is a prime. The function $f_u(x)=x^{\frac{p^n+3}{2}}+ux^2$ over $\gf_{p^n}[x]$ with $u\in\gf_{p^n}\setminus\{0,\pm1\}$ was recently studied by Budaghyan and Pal in \cite{Budaghyan2024ArithmetizationorientedAP}, whose differential uniformity is at most $5$ when $p^n\equiv3~(mod~4)$. In this paper, we study the differential uniformity and the differential spectrum of $f_u$ for $u=\pm1$. We first give some properties of the differential spectrum of any cryptographic function. Moreover, by solving some systems of equations over finite fields, we express the differential spectrum of $f_{\pm1}$ in terms of the quadratic character sums.
	\end{abstract}
	
	
	
\section{Introduction}
Let $\gf_{q}$ be the finite field with $q$ elements, where $q$ is a prime power. We denote by $\gf_{q}^*:=\gf_q \setminus \{0\}$. Any cryptographic function $F: \gf_{q} \rightarrow \gf_{q}$ can be uniquely represented as a univariate polynomial of degree less than $q$. For a function $F$, the main tools to study $F$ regarding the differential attack \cite{BS91} are the difference distribution table (DDT for short) and the differential uniformity introduced by Nyberg \cite{N94} in 1994. The DDT entry at point $(a,b)$ for any $a,b\in \gf_q$, denoted by $\delta_F(a,b)$, is defined as
\[\delta_F(a,b)=\#\{x\in\gf_q:\mathbb{D}_aF(x)=b\},\]
where $\mathbb{D}_a F(x)=F(x+a)-F(x)$ is the \textit{derivative function} of $F$ at the element $a$. Note that when $a=0$ and $b=0$, the equation $\mathbb{D}_a F(x)=b$ has $q$ solutions in $\gf_q$, which means $\delta_F(0,0)=q$. Besides, when $a=0$ and $b\in\gf_{p^n}^*$, the equation $\mathbb{D}_a F(x)=b$ has no solutions, which means $\delta_F(a,b)=0$. Therefore, for any polynomial, the DDT entries in the line $a=0$ are trivial. The differential uniformity of $F$, denoted by $\Delta_F$, is defined as
\[\Delta_F=\mathrm{max}\left\{\delta_F(a,b):a\in\gf_q^*, b\in\gf_q\right\}.\]
Generally speaking, the smaller the value of $\Delta_F$, the stronger the resistance of $F$ used in S-boxes against the differential attack. A cryptographic function $F$ is called differentially $k$-uniform if $\Delta_F=k$. Particularly when $\Delta_F=1$, $F$ is called a planar function \cite{DO68} or a perfect nonlinear (abbreviated as PN) function \cite{N91}. When $\Delta_F=2$, $F$ is called an almost perfect nonlinear (abbreviated as APN) function \cite{N94}, which is of the lowest possible differential uniformity over $\gf_{2^n}$ as in such finite fields, no PN functions exist. Readers may refer to \cite{BCC10}, \cite{BCCe20}, \cite{DMMe03}, \cite{HRS99}, \cite{HS97}, \cite{ORG}, \cite{PTW16}, \cite{XCX16}, \cite{ZHSS14}, \cite{ZW11} and references therein for some of the new developments on PN and APN functions. Apart from the concepts of PN and APN, a power function $F$ over $\gf_{p^n}$ is said to be locally-APN if \[\mathrm{max}\left\{\delta_F(1,b)| b\in\gf_{p^n}\setminus\gf_p\right\}=2.\] 
This definition was first introduced in \cite{HLXe23} for the case $p=2$ and generalized in \cite{XML23} for odd $p$. For a general function $F$, we can also give the concept of locally APN.
\begin{definition} Let $F$ be a function defined on $\gf_{p^n}$. Then $F$ is called locally-APN if
	\[\mathrm{max}\left\{\delta_F(a,b)|a\in \gf_{p^n}^*, b\in\gf_{p^n}\setminus\gf_p\right\}=2.\] 
\end{definition}


In \cite{BCC10}, the concept of the differential spectrum of a power function was introduced. The differential spectrum of a cryptographic function, compared with the differential uniformity, provides much more detailed information. In particular, the value distribution of the DDT is given directly by the differential spectrum. What's more, the differential spectrum has many applications such as in sequences \cite{BCC05}, \cite{DHKM01}, coding theory \cite{CCZ98}, \cite{CP19}, combinatorial design \cite{TDX20} etc. However, to determine the differential spectrum of a cryptographic function is usually a difficult problem.  Power functions with known differential spectra are summarized in Table \ref{tab:DS}.

For a polynomial function that is not a power function, the investigation of its differential spectrum is much more difficult. There are only a few cryptographic functions whose differential spectra were known \cite{LgJ24}, \cite{PdSW17}, \cite{FurXia}, \cite{WOS:001352564200001}. One of the focus of this paper is to explore a class of binomials studied in \cite{Budaghyan2024ArithmetizationorientedAP}. In \cite{Budaghyan2024ArithmetizationorientedAP}, the differential uniformity of $f_u(x)=x^{\frac{p^n+3}{2}}+ux^2$ with $u\in\gf_{p^n}\setminus\{0,\pm1\}$ has been investigated. In this paper, we determine the differential spectrum of such $f_u(x)$ when $u\in\{\pm1\}$, that is, $f_{\pm1}(x)=x^{\frac{p^n+3}{2}}\pm x^2$. 

This paper is organized as follows. Section \ref{sec:quadraticsums} presents certain quadratic character sums that are essential for the computation of the aimed differential spectrum. In Section \ref{sec:omega_i_eqts}, properties of the differential spectrum of any function are given. In Section \ref{sec:eqtsys_sols}, the number of solutions of several systems of equations over finite fields are investigated, which will be used in Section \ref{sec:difspm}, in which the differential spectrum of $f_{\pm1}$ is determined. Section \ref{sec:con} concludes this paper.

\begin{table}[!t]
	\footnotesize
	\centering
	\caption{Power Functions over $\gf_{p^n}$ with Known Differential Spectra}
	\begin{tabular}{|c|c|c|c|c|}
		\hline
		$p$ & $d$ & Condition & $\Delta_F$ & Ref \\
		\hline
		
		$2$ & $2^t+1$ & gcd$(t,n)=s$ & $2^s$ & \cite{BCC10} \\
		\hline

		$2$ & $2^{2t}-2^t+1$ & gcd$(t,n)=s,\frac{n}{s}\,odd$ & $2^s$ & \cite{BCC10} \\
		\hline

		$2$ & $2^n-2$ & $n\geqslant2$ & $2\,or\,4$ & \cite{BCC10} \\
		\hline

		$2$ & $2^{2k}+2^k+1$ & $n=4k$ & $4$ & \cite{BCC10},\cite{XY17} \\
		\hline

		$2$ & $2^t-1$ & $t=3,n-2$ & $6$ or $8$ & \cite{BCC11} \\
		\hline

		$2$ & $2^t-1$ & $t=\frac{n-1}{2}$, $\frac{n+3}{2}$, $n$ odd & $6$ or $8$ & \cite{BP14} \\
		\hline

		$2$ & $2^m+2^{(m+1)/2}+1$ & $n=2m$, $m\geqslant5$ odd & $8$ & \cite{XYY18} \\
		\hline

		$2$ & $2^{m+1}+3$ & $n=2m$, $m\geqslant5$ odd & $8$ & \cite{XYY18} \\
		\hline

		$2$ & $2^{3k}+2^{2k}+2^k-1$ & $n=4k$ & $2^{2k}$ & \cite{TuLe23} \\
		\hline
		
		$2$ & $\frac{2^m-1}{2^k+1}+1$ & $n=2m$, gcd$(k,m)=1$ & \makecell{$2^m$\\(locally APN)} & \cite{XML23} \\
		\hline

		$3$ & $2\cdot3^{(n-1)/2}+1$ & $n$ odd & $4$ & \cite{DHKM01} \\
		\hline

		$3$ & $\frac{3^n-1}{2}+2$ & $n$ odd & $4$ & \cite{JLLQ22} \\
		\hline

		$5$ & $\frac{5^n-3}{2}$ & any $n$ & $4$ or $5$ & \cite{YL21} \\
		\hline

		$5$ & $\frac{5^n+3}{2}$ & any $n$ & $3$ & \cite{PLZ23} \\
		\hline

		$p$ odd & $p^{2k}-p^k+1$ & gcd$(n,k)=e$, $\frac{n}{e} odd$ & $p^e+1$ & \cite{YZWe19}, \cite{LRF21} \\
		\hline

		$p$ odd & $\frac{p^k+1}{2}$ & gcd$(n,k)=e$ & $\frac{p^e-1}{2}$ or $p^e+1$ & \cite{CHNe13} \\
		\hline

		$p$ odd & $\frac{p^n+1}{p^m+1}+\frac{p^n-1}{2}$ & $p\equiv3\,(mod\;4)$, $m|n$, $n$ odd & $\frac{p^m+1}{2}$ & \cite{CHNe13} \\
		\hline

		$p$ odd & $p^n-3$ & any $n$ & $\leqslant5$ & \cite{XZLH18}, \cite{YXe22} \\
		\hline

		$p$ odd & $p^m+2$ & $n=2m$ & $2$ or $4$ & \cite{HRS99}, \cite{MXLH22} \\
		\hline

		$p$ odd & $2p^{\frac{n}{2}}-1$ & $n$ even & $p^{\frac{n}{2}}$  & \cite{YL22} \\
		\hline

		$p$ odd & $\frac{p^n-3}{2}$ & $p^n\equiv3\,(mod\;4)$, $p^n\geqslant7$ and $p^n\neq27$ & $2$ or $3$ & \cite{YMT24} \\
		\hline

		$p$ odd & $\frac{p^n+3}{2}$ & $p\geqslant5$, $p^n\equiv1\,(mod\;4)$ & $3$ & \cite{JLLQ21} \\
		\hline

		$p$ odd & $\frac{p^n+3}{2}$ & $p^n=11$ or $p^n\equiv3\,(mod\;4)$, $p\neq3$, $p^n\neq11$ & $2$ or $4$ & \cite{YMT23} \\
		\hline

		$p$ odd & $\frac{p^n+1}{4}$ & $p\neq3$, $p^n>7$, $p^n\equiv7\,(mod\;8)$& $2$ & \cite{TY23}, \cite{HRS99} \\
		\hline

		$p$ odd & $\frac{3p^n-1}{4}$ & $p\neq3$, $p^n>7$, $p^n\equiv3\,(mod\;8)$& $2$ & \cite{TY23}, \cite{HRS99} \\
		\hline

		$p$ odd & $\frac{p^n+1}{4}$, $\frac{3p^n-1}{4}$ & $p=3$ or $p>3$, $p^n\equiv3\,(mod\;4)$& $4$ & \cite{BXe23} \\
		\hline
		
		any $p$ & $k(p^m-1)$ & $n=2m$, gcd$(k,p^m+1)=1$ & \makecell{$p^m-2$\\(locally APN)} & \cite{HLXe23} \\
		\hline
	\end{tabular}
	\label{tab:DS}
\end{table}

\section{On Quadratic Character Sums}\label{sec:quadraticsums}
In this section, we will introduce some results on the quadratic character sums over the finite field $\gf_q$. Let $\chi(\cdot)$ be the quadratic multiplicative character of $\gf_q$, which is defined as
$$\chi(x)=
\left\{
\begin{array}{cl}
	1, & \mbox{if } x \mbox{ is a square in } \gf_q^*,\\
	-1,& \mbox{if } x \mbox{ is a nonsquare in }\gf_q^*, \\
	0, & \mbox{if } x=0.
\end{array}
\right.
$$
Let $\gf_q[x]$ be the polynomial ring over $\gf_q$. We consider the character sum of the form
\begin{eqnarray} \label{2:x1} \sum_{x\in\gf_q}\chi(f(x))\end{eqnarray}
with $f\in\gf_q[x]$. The case of $\deg(f) =1$ is trivial, and for $\deg(f) =2$, the following explicit formula was established in \cite{FF}.
\begin{lemma}\cite[Theorem 5.48]{FF}\label{lem:qsum_deg2}
	Let $f(x)=a_2x^2+a_1x+a_0\in\gf_q[x]$ with $p$ odd and $a_2\neq 0$. Put $d=a_1^2-4a_0a_2$ and let $\chi(\cdot)$ be the quadratic character of $\gf_q$. Then
	$$\sum\limits_{x\in\gf_q}\chi(f(x))=
	\left\{
	\begin{array}{lcl}
		-\chi(a_2), & if\;d\neq 0,\\
		(p^n-1)\chi(a_2), & if\;d=0. \\
	\end{array}
	\right.
	$$
\end{lemma}
Nevertheless, for a polynomials $f$ with degree $3$ or higher, computing $\sum\limits_{x\in\gf_q}\chi(f(x))$ or deriving a specific formula thereof is generally challenging. The subsequent lemma provides lower and upper bounds for any multiplicative character sum.
\begin{lemma}\cite{MR2514094}\label{lem:WeilB}
	Let $\gf_q$ be a finite field with $q$ odd. Let $f(x)=ax^3+bx^2+cx+d\in \gf_q[x]$ be a cubic polynomial with distinct roots in $\overline{\gf}_q$ and $\chi$ be a multiplicative character sum of $\gf_q$. Then we have
	$$\left|\sum_{x\in\gf_q}\chi\left(f(x)\right)\leqslant2\sqrt{q}\right|.$$
\end{lemma}

However, sometimes we need the exact value of the character sum $\sum\limits_{x\in\gf_q}\chi(f(x))$. For the case $deg(f(x))=3$, such a sum can be calculated by considering $\gf_{p^n}$-rational points of elliptic curves over $\gf_p$. More precisely, assume that $f$ is a cubic function over $\gf_{p^n}$ and denote \[\Gamma_{p,n}=\sum\limits_{x\in\gf_{p^n}}\chi(f(x)).\]
To evaluate $\Gamma_{p,n}$, several primary concepts from the theory of elliptic curves shall be taken into consideration. More details on the terminologies and notation can be found in \cite{MR2514094}. Let $E/\gf_p$ be the elliptic curve $E:y^2=f(x)$ over $\gf_{p^n}$, and $N_{p,n}$ denote the number of $F_{p^n}$-rational points (with the extra point at infinity) on the curve $E/\gf_p$.  From Subsection 1.3 and Theorem 2.3.1 in \cite[Chapter V]{MR2514094}, $N_{p,n}$ can be assessed by $\Gamma_{p,n}$. To be more exact, for every $n \geqslant 1$, $$N_{p,n} = p^n +1+\Gamma_{p,n}.$$
Furthermore, 
\begin{equation}\label{eqt:gammapn}
	\Gamma_{p,n}=-\alpha^n-\beta^n,
\end{equation}
where $\alpha$ and $\beta$ are the complex solutions of the quadratic equation $T^2 + \Gamma_{p,1}T + p = 0$. With an exploration, $\Gamma_{p,n}$ can be determined by $\Gamma_{p,1}$ directly and explicitly. We have
\begin{equation}\label{eqt:gammapn_fml}
\Gamma_{p,n}=\frac{(-1)^{n+1}}{2^{n-1}}\sum_{k=0}^{\lfloor\frac{n}{2}\rfloor}(-1)^k\binom{n}{2k}(\Gamma_{p,1})^{n-2k}(4p-(\Gamma_{p,1})^2)^k.
\end{equation}
Moreover, when $\Gamma_{p,1}=0$, we have
\begin{equation}\label{eqt:gammapn_0}
\Gamma_{p,n}=\left\{
	\begin{array}{ll}
	   (-1)^{\frac{n}{2}+1}\cdot2\cdot p^{\frac{n}{2}}, & n~\text{is even};\\
	   0, & n~\text{is odd}.\\
	\end{array}
\right.
\end{equation}
We define a specific character sum
$$\lambda_{p,n}=\sum\limits_{x\in\gf_{p^n}}\chi(x(x^2-2x-1)),$$
which will be used in the sequel. 
In the following examples, we give the exact value of $\lambda_{p,n}$ for a given $p$. 
\begin{example}
	Let $p=7$. For $n=1$, one has $\lambda_{7,1}=-4$. By (\ref{eqt:gammapn_fml}), we have
	$$\lambda_{7,n}=\sum_{k=0}^{\lfloor\frac{n}{2}\rfloor}(-1)^{k+1}\binom{n}{2k}2^{n-2k+1}3^k.$$
\end{example}
\begin{example}
	Let $p=31$. For $n=1$, one has $\lambda_{31,1}=0$. By (\ref{eqt:gammapn_0}), we have 
	$$\lambda_{31,n}=\left\{
		\begin{array}{ll}
		   (-1)^{\frac{n}{2}+1}\cdot2\cdot 31^{\frac{n}{2}}, & n~\text{is even};\\
		   0, & n~\text{is odd}.\\
		\end{array}
	\right.
	$$
\end{example}
The key to determine $\lambda_{p,n}$ is to calculate the value of $\lambda_{p,1}$. For the convenience, we list the values of $\lambda_{p,1}$ for all primes $3\leqslant p\leqslant1039$ with $p\equiv3~(mod~4)$ in Table \ref{tab:valuesoflambda}, which are computed by MAGMA.
\begin{table}[!htp]
	\footnotesize
	\centering
	\caption{The values of $\lambda_{p,1}$ for all primes $3\leqslant p\leqslant1039$ with $p\equiv3~(mod~4)$}
	\begin{tabular}{|c|c|c|c|c|c|c|c|c|c|c|c|c|c|c|c|}
		\hline
		$p$ & $3$ & $7$ & $11$ & $19$ & $23$ & $31$ & $43$ & $47$ & $59$ & $67$ & $71$ & $79$ & $83$ & $103$ & $107$ \\
		\hline
		
		$\lambda_{p,1}$ & $2$ & $-4$ & $-2$ & $2$ & $4$ & $0$ & $6$ & $-8$ & $14$ & $10$ & $12$ & $-8$ & $-6$ & $-4$ & $-2$ \\
		\hline

		\hline
		$p$ & $127$ & $131$ & $139$ & $151$ & $163$ & $167$ & $179$ & $191$ & $199$ & $211$ & $223$ & $227$ & $239$ & $251$ & $263$ \\
		\hline
		
		$\lambda_{p,1}$ & $16$ & $-6$ & $-10$ & $4$ & $2$ & $-20$ & $-6$ & $-16$ & $-4$ & $-22$ & $0$ & $18$ & $-24$ & $-18$ & $12$ \\
		\hline

		\hline
		$p$ & $271$ & $283$ & $307$ & $311$ & $331$ & $347$ & $359$ & $367$ & $379$ & $383$ & $419$ & $431$ & $439$ & $443$ & $463$ \\
		\hline
		
		$\lambda_{p,1}$ & $8$ & $6$ & $18$ & $-28$ & $14$ & $-18$ & $-4$ & $-8$ & $-2$ & $0$ & $26$ & $40$ & $36$ & $6$ & $-8$ \\
		\hline 
 
		\hline
		$p$ & $467$ & $479$ & $487$ & $491$ & $499$ & $503$ & $523$ & $547$ & $563$ & $571$ & $587$ & $599$ & $607$ & $619$ & $631$ \\
		\hline
		
		$\lambda_{p,1}$ & $-14$ & $0$ & $-20$ & $-10$ & $-22$ & $20$ & $14$ & $-38$ & $18$ & $38$ & $-34$ & $-12$ & $-16$ & $46$ & $-44$ \\
		\hline 

		\hline
		$p$ & $643$ & $647$ & $659$ & $683$ & $691$ & $719$ & $727$ & $739$ & $743$ & $751$ & $787$ & $811$ & $823$ & $827$ & $839$ \\
		\hline
		
		$\lambda_{p,1}$ & $42$ & $12$ & $-6$ & $-42$ & $-6$ & $24$ & $-12$ & $18$ & $44$ & $8$ & $-22$ & $-18$ & $-28$ & $22$ & $-36$ \\
		\hline 

		\hline
		$p$ & $859$ & $863$ & $883$ & $887$ & $907$ & $911$ & $919$ & $947$ & $967$ & $971$ & $983$ & $991$ & $1019$ & $1031$ & $1039$ \\
		\hline

		$\lambda_{p,1}$ & $-50$ & $-32$ & $34$ & $36$ & $38$ & $-24$ & $36$ & $-14$ & $28$ & $38$ & $20$ & $16$ & $6$ & $-20$ & $40$\\
		\hline

	\end{tabular}
	\label{tab:valuesoflambda}
\end{table}

At last, we present several results below concerning the exact values of specific character sums used in Section \ref{sec:difspm}.
\begin{lemma}\label{lem:qsum_ds}
	When $p^n\equiv3~(mod~4)$, we have
	\begin{enumerate}
		\item $\sum\limits_{x\in\gf_{p^n}}\chi\left(x(x-\frac{1}{2})(x-1)\right)=0$.
		\item $\sum\limits_{x\in\gf_{p^n}}\chi\left(\left(x-\frac{1}{2}\right)\left(x^2-x+\frac{1}{2}\right)\right)=0$.
		\item $\sum\limits_{x\in\gf_{p^n}}\chi\left(x(x-1)\left(x^2-x+\frac{1}{2}\right)\right)=-1$.
		\item $\sum\limits_{x\in\gf_{p^n}}\chi\left((x-1)\left(x^2-x+\frac{1}{2}\right)\right)=\lambda_{p,n}$.
		\item $\sum\limits_{x\in\gf_{p^n}}\chi\left(x\left(x-\frac{1}{2}\right)\left(x^2-x+\frac{1}{2}\right)\right)=-1-\chi(2)\lambda_{p,n}$.
	\end{enumerate}
\end{lemma}
\begin{proof}
	\begin{enumerate}
		\item Set $y=x-\frac{1}{2}$, then $x=y+\frac{1}{2}$ and 
				$$\sum\limits_{x\in\gf_{p^n}}\chi\left(x\left(x-\frac{1}{2}\right)(x-1)\right)=\sum\limits_{y\in\gf_{p^n}}\chi\left(\left(y+\frac{1}{2}\right)y\left(y-\frac{1}{2}\right)\right)=\sum\limits_{y\in\gf_{p^n}}\chi\left(y\left(y^2-\frac{1}{4}\right)\right).$$
		Let $y=-z$, then
				\[\sum\limits_{y\in\gf_{p^n}}\chi\left(y\left(y^2-\frac{1}{4}\right)\right)	=\sum\limits_{z\in\gf_{p^n}}\chi\left(-z\left(z^2-\frac{1}{4}\right)\right)=-\sum\limits_{z\in\gf_{p^n}}\chi\left(z\left(z^2-\frac{1}{4}\right)\right).\]
		Then $\sum\limits_{y\in\gf_{p^n}}\chi\left(y\left(y^2-\frac{1}{4}\right)\right)=0$. This implies that $\sum\limits_{x\in\gf_{p^n}}\chi\left(x\left(x-\frac{1}{2}\right)(x-1)\right)=0$. 
		\item Set  $y=-x+1$. Then $x=-y+1$ and 
		\begin{align*}
		\sum\limits_{x\in\gf_{p^n}}\chi\left(\left(x-\frac{1}{2}\right)\left(x^2-x+\frac{1}{2}\right)\right)&=\sum\limits_{y\in\gf_{p^n}}\chi\left(\left(-y+1-\frac{1}{2}\right)\left((-y+1)^2-(-y+1)+\frac{1}{2}\right)\right)\\
		&=-\sum\limits_{y\in\gf_{p^n}}\chi\left(\left(y-\frac{1}{2}\right)\left(y^2-y+\frac{1}{2}\right)\right).
		\end{align*}
		Hence $\sum\limits_{x\in\gf_{p^n}}\chi\left(\left(x-\frac{1}{2}\right)\left(x^2-x+\frac{1}{2}\right)\right)=0$.
		\item Let $u=x^2-x$. For any $u\in\gf_{p^n}$, the number of $x$'s satisfying $x^2-x=u$ is $1+\chi(1+4u)$. Then
			\begin{align*}
				\sum\limits_{x\in\gf_{p^n}}\chi\left(x(x-1)\left(x^2-x+\frac{1}{2}\right)\right)&=
				\sum\limits_{u\in\gf_{p^n}}\chi\left(u\left(u+\frac{1}{2}\right)\right)(1+\chi(1+4u))\\
				&=\sum\limits_{u\in\gf_{p^n}}\chi\left(u\left(u+\frac{1}{2}\right)\right)+\sum\limits_{u\in\gf_{p^n}}\chi\left(u\left(u+\frac{1}{2}\right)(1+4u)\right)\\
				&=\sum\limits_{u\in\gf_{p^n}}\chi\left(u\left(u+\frac{1}{2}\right)\right)+\sum\limits_{u\in\gf_{p^n}}\chi\left(u\left(u+\frac{1}{4}\right)\left(u+\frac{1}{2}\right)\right).
			\end{align*}
			Note that 
			\begin{align*}
			\sum\limits_{u\in\gf_{p^n}}\chi\left(u\left(u+\frac{1}{4}\right)\left(u+\frac{1}{2}\right)\right)&=\sum\limits_{v\in\gf_{p^n}}\chi\left(\left(-v-\frac{1}{2}\right)\left(-v-\frac{1}{4}\right)(-v)\right)\\
			&=-\sum\limits_{v\in\gf_{p^n}}\chi\left(v\left(v+\frac{1}{4}\right)\left(v+\frac{1}{2}\right)\right).
			\end{align*}
			Then $\sum\limits_{u\in\gf_{p^n}}\chi\left(u\left(u+\frac{1}{4}\right)\left(u+\frac{1}{2}\right)\right)=0$. This with Lemma \ref{lem:qsum_deg2} shows that $$\sum\limits_{x\in\gf_{p^n}}\chi\left(x(x-1)\left(x^2-x+\frac{1}{2}\right)\right)=-1.$$
		
			\item Set $y=x-1$, then $\sum\limits_{x\in\gf_{p^n}}\chi\left((x-1)\left(x^2-x+\frac{1}{2}\right)\right)=\sum\limits_{y\in\gf_{p^n}}\chi\left(y\left(y^2+y+\frac{1}{2}\right)\right)=\sum\limits_{y\in\gf_{p^n}^*}\chi\left(\frac{y^2+y+\frac{1}{2}}{y}\right)$.\\ Let $t=\frac{y^2+y+\frac{1}{2}}{y}$, then we can obtain a quadratic equation
				$$y^2+(1-t)y+\frac{1}{2}=0,$$
			whose discriminant is $\Delta=t^2-2t-1$. For $t\in\gf_{p^n}$, the number of $y$ satisfying the above quadratic equation is $1+\chi(\Delta)$.
			Then $$\sum\limits_{y\in\gf_{p^n}^*}\chi\left(\frac{y^2+y+\frac{1}{2}}{y}\right)=\sum\limits_{t\in\gf_{p^n}}(1+\chi(\Delta))\chi(t)=\sum_{t\in\gf_{p^n}}\chi(t)+\sum\limits_{t\in\gf_{p^n}}\chi(t(t^2-2t-1))=\lambda_{p,n}.$$
			\item It is clear that $x^2-x+\frac{1}{2}\neq 0$, otherwise $\left(x-\frac{1}{2}\right)^2=-\left(\frac{1}{2}\right)^2$, which contradicts to $\chi(-1)=-1$. Then $$\sum\limits_{x\in\gf_{p^n}}\chi\left(x\left(x-\frac{1}{2}\right)\left(x^2-x+\frac{1}{2}\right)\right)=\sum\limits_{x\in\gf_{p^n}}\chi\left(\frac{x(x-\frac{1}{2})}{x^2-x+\frac{1}{2}}\right).$$
				Let $t=\frac{x(x-\frac{1}{2})}{x^2-x+\frac{1}{2}}$, then
				\begin{equation}\label{eqt:qsum_ds_t_2}
					\left(t-1\right)x^2+\left(\frac{1}{2}-t\right)x+\frac{t}{2}=0.
				\end{equation}
				Note that $x=1$ if and only if $t=1$. When $t\neq1$, the discriminant of (\ref{eqt:qsum_ds_t_2}) is $-t^2+t+\frac{1}{4}$. Then we have
				\begin{equation*}
					\begin{array}{ll}
						\sum\limits_{x\in\gf_{p^n}}\chi\left(\frac{x\left(x-\frac{1}{2}\right)}{x^2-x+\frac{1}{2}}\right)&=\chi(1)+\sum\limits_{t\in\gf_{p^n},t\neq1}\left(1+\chi(\Delta)\right)\chi(t) \\
						&=1+(-2)+\sum\limits_{t\in\gf_{p^n}}\left(1-\chi\left(t^2-t-\frac{1}{4}\right)\right)\chi(t) \\
						&=-1-\sum\limits_{t\in\gf_{p^n}}\chi\left(t\left(t^2-t-\frac{1}{4}\right)\right).
					\end{array}
				\end{equation*}
				Set $t=\frac{y}{2}$, then
				$$\sum\limits_{t\in\gf_{p^n}}\chi\left(t\left(t^2-t-\frac{1}{4}\right)\right)=\sum\limits_{t\in\gf_{p^n}}\chi\left(\frac{y}{2}\left(\frac{y^2}{4}-\frac{y}{2}-\frac{1}{4}\right)\right)=\chi(2)\lambda_{p,n}.$$
				That is desired result follows.
	\end{enumerate}
\end{proof}

\section{The properties of the differential spectrum of a general cryptographic function over finite field}\label{sec:omega_i_eqts}
First, we give the definition of the differential spectrum of a cryptographic function. 
\begin{definition}\label{def:ds}
Let $F$ be a function from $\gf_{q}$ to $\gf_{q}$ with differential uniformity $\Delta_F$, and
$$ \omega_i = \#\{(a,b) \in \gf_{q} \times \gf_{q} : \delta_F(a,b) = i\},~0 \leqslant i \leqslant q,$$
where $\delta_F(a,b)=\#\{x\in\gf_q:F(x+a)-F(x)=b\}$. The differential spectrum of $F$ is defined as the multiset
$$\mathbb{S}_F =\left[\omega_0,\omega_1,...,\omega_{\Delta_F},\omega_{\Delta_F+1},\cdots,\omega_{q}\right].$$
\end{definition}

Sometimes we ignore the zeros in the differential spectrum. We remark that our definition of the differential spectrum is a little different from that in \cite{FurXia}. In our definition of $\omega_i$, we consider all the pairs $(a,b)\in \gf_{q} \times \gf_{q}$, including $a=0$. The values of $\omega_i (i>\Delta_F)$ can be obtained easily. That is, $\omega_{\Delta_F+1}=\cdots=\omega_{q-1}=0$ and $\omega_q=1$.  

From \cite{HRS99}, it is known that the differential spectrum of a power function satisfies several identities. It is natural to consider how it behaves with respect to the differential spectrum of any function. In this section, we give some identities of the differential spectrum of a general cryptographic function. Let $f$ be a polynomial over $\gf_q$ with differential uniformity $\Delta_f$.  We have the following theorem.
\begin{theorem}\label{thm:sol_eqts}
We have
	\begin{equation}
		\sum\limits_{i=0}^{q}\omega_i=q^2,
	\end{equation}
	\begin{equation}
		\sum\limits_{i=0}^{q}i\omega_i=q^2,
	\end{equation}
and
	\begin{equation}
		\sum\limits_{i=0}^{q}i^2\omega_i=N_4,
	\end{equation}
	where
	\begin{equation}
		N_4=\#\bigg\{(x_1,x_2,x_3,x_4)\in(\gf_q)^4:\left\{
			\begin{array}{l}
				x_1-x_2+x_3-x_4=0, \\
				f(x_1)-f(x_2)+f(x_3)-f(x_4)=0. \\
			\end{array}
			\right.\bigg\}.
	\end{equation}
\end{theorem}
\begin{proof}
	According to the definition of $\omega_i$, we have
	$$\sum\limits_{i=0}^{q}\omega_i=\sum\limits_{i=0}^{q}\#\{(a,b)\in(\gf_q)^2:\delta_f(a,b)=i\}=q^2.$$
	The last equation holds since when $i$ runs through the integers in the range $[0,q]$, each $(a,b)\in(\gf_q)^2$ should occur.
	
	Besides, for a fixed $i$, there are $\omega_i$ distinct pairs $(a,b)$ such that $\delta_f(a,b)=i$, then we have
	$$i\omega_i=\sum\limits_{\delta_f(a,b)=i}\#\{x\in\gf_q:f(x+a)-f(x)=b\}.$$
	And for a given $a\in\gf_q$, \[\sum\limits_{b\in\gf_q}\#\{x\in\gf_q:f(x+a)-f(x)=b\}=\sum\limits_{x\in\gf_q}1=q\]
	since for any $x\in\gf_q$, there exists exactly one $b\in\gf_{q}$ satisfying $f(x+a)-f(x)=b$.
	Then
	\begin{align*}
		\sum\limits_{i=0}^{q}i\omega_i&=\sum\limits_{i=0}^{q}\sum\limits_{\substack{(a,b)\in(\gf_q)^2\\\delta_f(a,b)=i}}\#\{x\in\gf_q:f(x+a)-f(x)=b\}\\
		&=\sum\limits_{(a,b)\in(\gf_q)^2}\#\{x\in\gf_q:f(x+a)-f(x)=b\}\\
		&=\sum\limits_{a\in\gf_q}\sum\limits_{b\in\gf_q}\#\{x\in\gf_q:f(x+a)-f(x)=b\}\\
		&=\sum\limits_{a\in\gf_q}\sum\limits_{x\in\gf_q}1\\
		&=q^2.
	\end{align*}
	In the following, we prove the last statement. Note that
	$$\delta_f(\alpha,\beta)=\#\bigg\{(x_1,x_2)\in(\gf_q)^2:\left\{
		\begin{array}{l}
			x_1-x_2=\alpha, \\
			f(x_1)-f(x_2)=\beta. \\
		\end{array}
		\right.\bigg\},$$
since the number of solutions of $x_2$ of the equation $f(x_2+\alpha)-f(x_2)=\beta$ is $\delta_f(\alpha,\beta)$ and $x_1$ is uniquely determined by $x_2$.

	It is clear that
	\begin{align*}
		N_4&=\#\bigg\{(x_1,x_2,x_3,x_4)\in(\gf_q)^4:\left\{
			\begin{array}{l}
				x_1-x_2+x_3-x_4=0, \\
				f(x_1)-f(x_2)+f(x_3)-f(x_4)=0. \\
			\end{array}
			\right.\bigg\}\\
			&=\sum\limits_{(\alpha,\beta)\in(\gf_q)^2}\#\bigg\{(x_1,x_2,x_3,x_4)\in(\gf_q)^4:\left\{
				\begin{array}{l}
					x_1-x_2=x_4-x_3=\alpha, \\
					f(x_1)-f(x_2)=f(x_4)-f(x_3)=\beta. \\
				\end{array}
				\right.\bigg\}\\
			&=\sum\limits_{(\alpha,\beta)\in(\gf_q)^2}\#\bigg\{(x_1,x_2)\in(\gf_q)^2:\left\{
				\begin{array}{l}
					x_1-x_2=\alpha, \\
					f(x_1)-f(x_2)=\beta. \\
				\end{array}
				\right.\bigg\}\cdot\#\bigg\{(x_3,x_4)\in(\gf_q)^2:\left\{
				\begin{array}{l}
					x_4-x_3=\alpha, \\
					f(x_4)-f(x_3)=\beta. \\
				\end{array}
				\right.\bigg\}\\
			&=\sum\limits_{(\alpha,\beta)\in(\gf_q)^2}(\delta_f(\alpha,\beta))^2\\
			&=\sum\limits_{i=0}^{q}\sum\limits_{\substack{(\alpha,\beta)\in(\gf_q)^2\\\delta_f(\alpha,\beta)=i}}i^2\\
			&=\sum\limits_{i=0}^{q}i^2\omega_i.\\
	\end{align*}
	This completes the proof.
\end{proof}

\section{On the number of solutions of certain systems of equations}\label{sec:eqtsys_sols}
In this section, we determine the number of solutions of several systems of equations which are needed in Section \ref{sec:difspm}.

\begin{lemma}\label{lem:eqtsys_y_111}
	Let $p^n\equiv~3(mod~4)$. Let $\dot{N}_{(1,1,1)}$ denote the number of solutions $(y_1,y_2,y_3)\in(\gf_{p^n}^*)^3$ of the following system of equations
	\begin{equation}\label{eqt:eqtsys_y_111}
		\left\{
			\begin{array}{l}
				y_1-y_2+y_3-1=0,\\
				y_1^2-y_2^2+y_3^2-1=0, \\
			\end{array}
		\right.
	\end{equation}
	with $(\chi(y_1),\chi(y_2),\chi(y_3)) = (1,1,1)$. Then we have $\dot{N}_{(1,1,1)}=p^n-2$.
\end{lemma}
\begin{proof}
	The system (\ref{eqt:eqtsys_y_111}) can be rewritten as 
	\begin{equation}\label{eqt:eqtsys_y_111_2}
	\left\{
		\begin{array}{l}
			y_1-y_2=1-y_3,\\
			y_1^2-y_2^2=1-y_3^2. \\
		\end{array}
	\right.
	\end{equation}
	If $y_3=1$, then we get $(y_1,y_2,y_3)=(y_2,y_2,1)$. Note that $(y_2,y_2,1)$ is a desired solution if and only if $\chi(y_2)=1$. Therefore, the number of such desired solutions is $\frac{p^n-1}{2}$. If $y_3\neq1$, then $y_1\neq y_2$, we have
	$$
	\left\{
		\begin{array}{l}
			y_1-y_2=1-y_3,\\
			y_1+y_2=1+y_3, \\
		\end{array}
	\right.
	$$
	whose solution is $(y_1,y_2,y_3)=(1,y_2,y_2)$. Similarly, the number of such desired solutions is $\frac{p^n-1}{2}$. Together with the two cases and removing one identical solution $(1,1,1)$, $\dot{N}_{(1,1,1)}=\frac{p^n-1}{2}+\frac{p^n-1}{2}-1=p^n-2$.
\end{proof}

\begin{lemma}\label{lem:eqtsys_y_m1m1m1}
	Let $p^n\equiv~3(mod~4)$. Let $\ddot{N}_{(-1,-1,-1)}$ denote the number of solutions $(y_1,y_2,y_3)\in(\gf_{p^n}^*)^3$ of the following system of equations
	\begin{equation}\label{eqt:eqtsys_y_m1m1m1}
		\left\{
			\begin{array}{l}
				y_1-y_2+y_3-1=0,\\
				y_1^2-y_2^2+y_3^2=0, \\
			\end{array}
		\right.
	\end{equation}
	with $(\chi(y_1),\chi(y_2),\chi(y_3)) = (-1,-1,-1)$. Then we have $\ddot{N}_{(-1,-1,-1)}=\frac{1}{8}\left(p^n+1+(\chi(2)-1)\lambda_{p,n}\right)$.
\end{lemma}
\begin{proof}
	It is easy to check that $y_3\neq1$ in (\ref{eqt:eqtsys_y_m1m1m1}), then we have 
	$$\left\{
		\begin{array}{l}
			y_1-y_2=1-y_3,\\
			y_1+y_2=-\frac{y_3^2}{1-y_3}. \\
		\end{array}
	\right.
	$$
	Thus, we obtain the solutions of (\ref{eqt:eqtsys_y_m1m1m1})
	$$\left\{
		\begin{array}{l}
			y_1=1+\frac{1}{2(y_3-1)},\\
			y_2=y_3+\frac{1}{2(y_3-1)}. \\
		\end{array}
	\right.
	$$
	Hence, $(y_1,y_2,y_3)$ is a desired solution if and only if 
	$$\chi\left(1+\frac{1}{2(y_3-1)}\right)=-1,~\chi\left(y_3+\frac{1}{2(y_3-1)}\right)=-1,~\chi(y_3)=-1.$$
	This implies that 
	\begin{align*}
		\ddot{N}_{(-1,-1,-1)}&=\frac{1}{8}\sum\limits_{y_3\in\gf_{p^n}^*,y_3\neq1}\left[1-\chi\left(1+\frac{1}{2(y_3-1)}\right)\right]\cdot\left[1-\chi\left(y_3+\frac{1}{2(y_3-1)}\right)\right]\cdot\left[1-\chi(y_3)\right]\\
		&=\frac{1}{8}\sum\limits_{y_3\in\gf_{p^n}^*,y_3\neq1}\left[1-\chi\left((y_3-1)\left(y_3-\frac{1}{2}\right)\right)\right]\left[1-\chi\left((y_3-1)\left(y_3^2-y_3+\frac{1}{2}\right)\right)\right][1-\chi(y_3)]\\
		&=\frac{1}{8}\Bigg[\sum\limits_{y_3\in\gf_{p^n}}1-\sum\limits_{y_3\in\gf_{p^n}}\chi(y_3)-\sum\limits_{y_3\in\gf_{p^n}}\chi\left((y_3-1)\left(y_3-\frac{1}{2}\right)\right)-\sum\limits_{y_3\in\gf_{p^n}}\chi\left((y_3-1)\left(y_3^2-y_3+\frac{1}{2}\right)\right) \\
		&~~~~~~~+\sum\limits_{y_3\in\gf_{p^n}}\chi\left((y_3-1)^2\left(y_3-\frac{1}{2}\right)\left(y_3^2-y_3+\frac{1}{2}\right)\right)+\sum\limits_{y_3\in\gf_{p^n}}\chi\left(y_3(y_3-1)\left(y_3^2-y_3+\frac{1}{2}\right)\right)\\
		&~~~~~~~+\sum\limits_{y_3\in\gf_{p^n}}\chi\left(y_3(y_3-1)\left(y_3-\frac{1}{2}\right)\right)-\sum\limits_{y_3\in\gf_{p^n}}\chi\left(y_3\left(y_3-\frac{1}{2}\right)(y_3-1)^2\left(y_3^2-y_3+\frac{1}{2}\right)\right)\Bigg]\\
		&=\frac{1}{8}\Bigg[p^n-\sum\limits_{y_3\in\gf_{p^n}}\chi\left((y_3-1)\left(y_3-\frac{1}{2}\right)\right)-\sum\limits_{y_3\in\gf_{p^n}}\chi\left((y_3-1)\left(y_3^2-y_3+\frac{1}{2}\right)\right) \\
		&~~~~~~~+\sum\limits_{y_3\in\gf_{p^n}}\chi\left(\left(y_3-\frac{1}{2}\right)\left(y_3^2-y_3+\frac{1}{2}\right)\right)-1+\sum\limits_{y_3\in\gf_{p^n}}\chi\left(y_3(y_3-1)\left(y_3^2-y_3+\frac{1}{2}\right)\right)\\
		&~~~~~~~+\sum\limits_{y_3\in\gf_{p^n}}\chi\left(y_3\left(y_3-\frac{1}{2}\right)(y_3-1)\right)-\sum\limits_{y_3\in\gf_{p^n}}\chi\left(y_3\left(y_3-\frac{1}{2}\right)\left(y_3^2-y_3+\frac{1}{2}\right)\right)+1\Bigg]\\
		&=\frac{1}{8}(p^n+1+(\chi(2)-1)\lambda_{p,n}).
	\end{align*}
	The last identity holds based on Lemma \ref{lem:qsum_deg2} and Lemma \ref{lem:qsum_ds}.
\end{proof}

\begin{lemma}\label{lem:eqtsys_y_xallmisus}
	Let $p^n\equiv~3(mod~4)$. Let $\dddot{N}_{(-1,-1,-1,-1)}$ denote the number of solutions $(y_1,y_2,y_3,y_4)\in(\gf_{p^n}^*)^4$ of the equation
	\begin{equation}\label{eqt:eqtsys_y_xallmisus}
		y_1-y_2+y_3-y_4=0,
	\end{equation}
	with $(\chi(y_1),\chi(y_2),\chi(y_3),\chi(y_4)) = (-1,-1,-1,-1)$. Then $\dddot{N}_{(-1,-1,-1,-1)}=\frac{1}{16}\left((p^n-1)\left(p^{2n}-2p^n+5\right)\right)$.
\end{lemma}
\begin{proof}
We have
	\begin{align*}
		&\#\left\{(y_1,y_2,y_3,y_4)\in(\gf_{p^n}^*)^4:y_1-y_2+y_3-y_4=0\right\} \\
		=&\#\left\{(y_1,y_2,y_3,y_4)\in(\gf_{p^n}^*)^4:y_1-y_2=y_4-y_3\right\} \\
		=&\sum\limits_{\alpha\in\gf_{p^n}}\#\Bigg\{(y_1,y_2,y_3,y_4)\in(\gf_{p^n}^*)^4:\left\{
			\begin{array}{l}
				y_1-y_2=\alpha, \\
				y_4-y_3=\alpha. 
			\end{array}
		\right.
		\Bigg\} \\
		=&\sum\limits_{\alpha\in\gf_{p^n}}\#\left\{(y_1,y_2)\in(\gf_{p^n}^*)^2:y_1-y_2=\alpha\right\}\cdot\#\left\{(y_3,y_4)\in(\gf_{p^n}^*)^2:y_4-y_3=\alpha\right\} \\
		=&\sum\limits_{\alpha\in\gf_{p^n}}\left(\#\left\{(y_1,y_2)\in(\gf_{p^n}^*)^2:y_1-y_2=\alpha\right\}\right)^2. \\
	\end{align*}
In the following, for a given $\alpha\in\gf_{p^n}$, we discuss the number of solutions $(y_1,y_2)$ of the equation 
\[y_1-y_2=\alpha\]
with $\chi(y_1)=\chi(y_2)=-1$.
	\begin{enumerate}
		\item $\alpha=0$. Then the desired solutions of $y_1-y_2=\alpha$ should be  $(y_1,y_2)=(y_1,y_1)$ with $\chi(y_1)=-1$. Besides, the number of such solutions is $\frac{p^n-1}{2}$.
		\item $\alpha$ is a square element in $\gf_{p^n}^*$. Let $z_i=\frac{y_i}{\alpha} (i=1,2)$, then $\chi(z_1)=\chi(z_2)=-1$ and the equation $y_1-y_2=\alpha$ becomes $z_1-z_2=1$. Thus, we have 
		\[\#\left\{(z_1,z_2)\in(\gf_{p^n}^*)^2:z_1-z_2=1,\chi(z_1)=\chi(z_2)=-1\right\}=\#\left\{z_2\in\gf_{p^n}^*:\chi(z_2+1)=\chi(z_2)=-1\right\}\]
		\[=\frac{1}{4}\sum\limits_{\substack{z_2\in\gf_{p^n}^*,\\z_2\neq-1}}(1-\chi(z_2))(1-\chi(z_2+1))=\frac{p^n-3}{4}.\]
		\item $\alpha$ is a nonsquare element in $\gf_{p^n}^*$. The number of solutions $(y_1,y_2)$ with $\chi(y_1)=\chi(y_2)=-1$ of the equation $y_1-y_2=\alpha$ is also $\frac{p^n-3}{4}$. The proof is similar with 2) and we omit it.
	\end{enumerate}
In summary, we have $$\dddot{N}_{(-1,-1,-1,-1)}=\left(\frac{p^n-1}{2}\right)^2+2\cdot\frac{p^n-1}{2}\cdot\left(\frac{p^n-3}{4}\right)^2=\frac{1}{16}\left((p^n-1)(p^{2n}-2p^n+5)\right).$$
\end{proof}

\section{The differential spectrum of $f_1$}\label{sec:difspm}
Let $n$ be an odd integer, $p$ be an odd prime satisfying $p\equiv3~(mod~4)$. Recall that  $f_1(x)=x^{\frac{p^n+3}{2}}+x^{2}$ and $f_{-1}(x)=x^{\frac{p^n+3}{2}}-x^{2}$ are binomials over $\gf_{p^n}$. Note that $f_{-1}(x)=-f_1(-x)$, then we only study the differential properties of $f_1(x)$. In this section, we investigate the differential uniformity and the differential spectrum of $f_1$. The differential uniformity and the differential spectrum of $f_{-1}$ can be obtained directly and we omit them.

\subsection{The differential uniformity of $f_1$}\label{sec:subsec1}
In this subsection, our primary objective is to determine the differential uniformity of $f_1$, accompanied by a discussion of the number of potential solutions associated with the differential equation.
\begin{theorem}\label{uniformity}
	Let $n$ be an odd integer, $p$ be an odd prime with $p^n\equiv3~(mod~4)$. The differential uniformity of $f_1$ is $\frac{p^n+1}{4}$. Moreover, $\delta_{f_1}(a,b)\leq 2$ when $(a,b)\in(\gf_{p^n}^*)^2$.
\end{theorem}
\begin{proof} It is obvious that $\delta_{f_1}(0,0)=p^n$ and $\delta_{f_1}(0,b)=0$ for $b\neq 0$. For any $(a,b)\in\gf_{p^n}^*\times\gf_{p^n}$, the differential equation $f_1(x+a)-f_1(x)=b$ becomes
\begin{equation}\label{eqt:diffeqt}
	(\chi(x+a)-\chi(x))x^2+2a(1+\chi(x+a))x+(1+\chi(x+a))a^2-b=0.
\end{equation}
When $x\notin\{0,-a\}$, we discuss (\ref{eqt:diffeqt}) in four cases shown in Table \ref{tab:FCases}, in which $x_1$ and $x_2$ denote the two solutions of the quadratic equations in Case III and Case IV. 
\begin{table}
    \centering
    \caption{List of Equations and Solutions}
    \begin{tabular}{|c|c|c|c|c|}
        \hline
        Case  & I &  II &  III &  IV \\
        \hline
        $(\chi(x+a),\chi(x))$ & $(1,1)$ & $(-1,-1)$ & $(-1,1)$ & $(1,-1)$\\
        \hline
        $Equation$ & $4ax+2a^2-b=0$ & $b=0$ & $2x^2+b=0$ & $2x^2+4ax+2a^2-b=0$ \\
        \hline
		$x$ & $\frac{b-2a^2}{4a}$ & $ $ & $\pm\sqrt{-\frac{b}{2}}$ & $\frac{-2a\pm\sqrt{2b}}{2}$ \\
		\hline
		$x+a$ & $\frac{b+2a^2}{4a}$ & $ $ & $a\pm\sqrt{-\frac{b}{2}}$ & $-a\pm a\sqrt{1+\frac{u-1}{ab}}$ \\
		\hline
		$x_1x_2$ & $ $ & $ $ & $\frac{b}{2}$ & $\frac{2a^2-b}{2}$ \\
		\hline
		$(x_1+a)(x_2+a)$ & $ $ & $ $ & $a^2+\frac{b}{2}$ & $-\frac{b}{2}$ \\
		\hline
    \end{tabular}
	\label{tab:FCases}
 \end{table}
\begin{enumerate}
	\item When $b=0$, $x=0$ is a solution of (\ref{eqt:diffeqt}) if and only if $\chi(a)=-1$, and $x=-a$ is a solution of equation (\ref{eqt:diffeqt}) if and only if $\chi(a)=1$. This indicates that for any $a\in\gf_{p^n}^*$, $f_1(x+a)-f_1(x)=0$ has exactly one solution in $\{0,-a\}$. For the remaining solutions not in $\{0,-a\}$, replacing $b$ by $0$ in Table \ref{tab:FCases}, it is effortless to check that $f_1(x+a)-f_1(x)=0$ has no solutions in Case I, Case III or Case IV. And for any $a\in\gf_{p^n}^*$, there are $$\frac{1}{4}\sum\limits_{x\in\gf_{p^n}^*\setminus\{-a\}}(1-\chi(x+a))(1-\chi(x))=\frac{p^n-3}{4}$$ solutions in Case II. In short, the equation $f_1(x+a)-f_1(x)=0$ has $\frac{p^n+1}{4}$ solutions in total and the number of such $(a,0)$ is $p^n-1$.
	\item When $b\neq0$, 
	\begin{enumerate}
		\item it is clear that (\ref{eqt:diffeqt}) cannot have solutions in both Cases III and IV simultaneously as $\chi(2b)$ cannot be both $-1$ and $1$ at the same time;
		\item if (\ref{eqt:diffeqt}) has two solutions in Case III, then $\chi(-\frac{b}{2})=1$ and $\chi(x_1x_2)=\chi(\frac{b}{2})=1$, which is a contradiction;
		\item (\ref{eqt:diffeqt}) has at most one solution in Case IV, otherwise both $\chi(2b)=1$ and $\chi(x_1+a)\chi(x_2+a)=\chi(-2b)=1$ would hold simultaneously, which is impossible with $\chi(-1)=-1$.
	\end{enumerate}
	From the discussion above, the equation $f_1(x+a)-f_1(x)=b$ has at most two solutions when $b\neq0$.
\end{enumerate}
We finish the proof. 
\end{proof}

\subsection{The value of $N_4$ pertaining to $f_1$}\label{sec:subsec2}
Based on the discussion in Subsection \ref{sec:subsec1}, to determine the differential spectrum of $f_1$, we are required to determine $\omega_0$, $\omega_1$ and $\omega_2$. Further, according to Theorem \ref{thm:sol_eqts} and the fact that $f_1(x)=(1+\chi(x))x^2$, we need to examine the solutions of the system of equations
	$$\left\{
		\begin{array}{l}
			x_1-x_2+x_3-x_4=0,\\
			(1+\chi(x_1))x_1^2-(1+\chi(x_2))x_2^2+(1+\chi(x_3))x_3^2-(1+\chi(x_4))x_4^2=0. \\
		\end{array}
		\right.
	$$

\begin{theorem}\label{thm:N4_f1}
	Let $N_4$ denote the number of solutions of the system of equations:
	\begin{equation}\label{eqt:eqtsys_org}
		\left\{
		\begin{array}{l}
			x_1-x_2+x_3-x_4=0,\\
			(1+\chi(x_1))x_1^2-(1+\chi(x_2))x_2^2+(1+\chi(x_3))x_3^2-(1+\chi(x_4))x_4^2=0. \\
		\end{array}
		\right.
	\end{equation}	
	Then $N_4=\frac{1}{16}\left((p^n-1)\left(p^{2n}+34p^n+17+4(\chi(2)-1)\lambda_{p,n}\right)\right)+1$.
\end{theorem}
\begin{proof}
For a solution $(x_1,x_2,x_3,x_4)\in(\gf_{p^n})^4$ of (\ref{eqt:eqtsys_org}), let $\mathcal{N}^{(i)}$ denote the number of solutions containing $i$ zeros, where $0\leqslant i\leqslant4$. In the first place, we are trying to evaluate $\mathcal{N}^{(0)}$. Let $\mathcal{N}_{(i,j,k,l)}$ denote the number of solutions $(x_1,x_2,x_3,x_4)\in(\gf_{p^n}^*)^4$ of the system (\ref{eqt:eqtsys_org}) when $(\chi(x_1),\chi(x_2),\chi(x_3),\chi(x_4))=(i,j,k,l), i,j,k,l\in\{\pm1\}$. Then we have $\mathcal{N}^{(0)}=\sum\limits_{i,j,k,l\in\{\pm1\}}\mathcal{N}_{(i,j,k,l)}$. Next, we compute $\mathcal{N}_{(i,j,k,l)}$ in $16$ cases presented below.
\begin{enumerate}
	\item When $(\chi(x_1),\chi(x_2),\chi(x_3),\chi(x_4))=(1,1,1,1)$, the system (\ref{eqt:eqtsys_org}) can be reduced to
		$$\left\{
		\begin{array}{l}
			x_1-x_2+x_3-x_4=0,\\
			x_1^2-x_2^2+x_3^2-x_4^2=0. \\
		\end{array}
		\right.
		$$
		Set $y_i=\frac{x_i}{x_4}$, then we need to calculate the number of solutions of 
		$$
			\left\{
				\begin{array}{l}
					y_1-y_2+y_3-1=0,\\
					y_1^2-y_2^2+y_3^2-1=0, \\
				\end{array}
			\right.
		$$
		where $\chi(y_1)=\chi(y_2)=\chi(y_3)=1$.
		According to Lemma \ref{lem:eqtsys_y_111}, the number of the above system is $p^n-2$. Combined with the condition that $\chi(x_4)=1$ and noting that the number of such $x_4$ is $\frac{p^n-1}{2}$, we can obtain that $\mathcal{N}_{(1,1,1,1)}=\frac{1}{2}(p^n-1)(p^n-2)$.
	\item In this case, we consider there is exactly one nonsquare element among $x_1,x_2,x_3$ and $x_4$. 
		\begin{enumerate}
			\item When $(\chi(x_1),\chi(x_2),\chi(x_3),\chi(x_4))=(1,1,1,-1)$, the system (\ref{eqt:eqtsys_org}) can be reduced to
			$$\left\{
			\begin{array}{l}
				x_1-x_2+x_3-x_4=0,\\
				x_1^2-x_2^2+x_3^2=0. \\
			\end{array}
			\right.
			$$
			Set $y_i=\frac{x_i}{x_4}$, then we need to calculate the number of solutions of 
			$$
				\left\{
					\begin{array}{l}
						y_1-y_2+y_3-1=0,\\
						y_1^2-y_2^2+y_3^2=0, \\
					\end{array}
				\right.
			$$
			with $(\chi(y_1),\chi(y_2),\chi(y_3))=(-1,-1,-1)$.
			According to Lemma \ref{lem:eqtsys_y_m1m1m1}, the number of solutions of the above system  is $\frac{1}{8}(p^n+1+(\chi(2)-1)\lambda_{p,n})$. Combined with the condition that $\chi(x_4) = -1$ and noting that the number of such $x_4$ is $\frac{p^n-1}{2}$, we can obtain that $\mathcal{N}_{(1,1,1,-1)}=\frac{1}{16}\left((p^n-1)(p^n+1+(\chi(2)-1)\lambda_{p,n})\right)$.
			\item When $(\chi(x_1),\chi(x_2),\chi(x_3),\chi(x_4))=(1,1,-1,1)$, the system (\ref{eqt:eqtsys_org}) can be reduced to
			$$\left\{
			\begin{array}{l}
				x_1-x_2+x_3-x_4=0,\\
				x_1^2-x_2^2-x_4^2=0. \\
			\end{array}
			\right.
			$$
			This system of equations is the same as
			$$\left\{
			\begin{array}{l}
				x_2-x_1+x_4-x_3=0,\\
				x_2^2-x_1^2+x_4^2=0. \\
			\end{array}
			\right.
			$$
			By a simple comparison, we have $\mathcal{N}_{(1,1,-1,1)}=\mathcal{N}_{(1,1,1,-1)}$. In the same manner, $\mathcal{N}_{(-1,1,1,1)}=\mathcal{N}_{(1,1,-1,1)}=\mathcal{N}_{(1,-1,1,1)}=\mathcal{N}_{(1,1,1,-1)}$ can be deduced.
		\end{enumerate}
		In short, we have $\mathcal{N}_{(1,1,1,-1)}=\mathcal{N}_{(1,1,-1,1)}=\mathcal{N}_{(1,-1,1,1)}=\mathcal{N}_{(-1,1,1,1)}=\frac{1}{16}\big((p^n-1)(p^n+1+(\chi(2)-1)\lambda_{p,n})\big)$.
	\item In this case, we consider there are exactly two nonsquare elements among $x_1,x_2,x_3$ and $x_4$. 
		\begin{enumerate}
		\item When $(\chi(x_1),\chi(x_2),\chi(x_3),\chi(x_4))=(1,1,-1,-1)$, the system (\ref{eqt:eqtsys_org}) can be reduced to
		$$\left\{
		\begin{array}{l}
			x_1-x_2+x_3-x_4=0,\\
			x_1^2-x_2^2=0. \\
		\end{array}
		\right.
		$$
		From the second equation above, we can obtain that $x_1=x_2$ since $\chi(x_1)=\chi(x_2)$ and $\chi(-1)=-1$. Then the solutions of this system of equations is $(x_1,x_1,x_3,x_3)$. Combined with the condition that $\chi(x_1)=1$ and $\chi(x_3)=-1$, and noting that the number of such $x_1$ and $x_3$ is each $\frac{p^n-1}{2}$, we can obtain that $\mathcal{N}_{(1,1,-1,-1)}=\frac{1}{4}(p^n-1)^2$. 
		\item Since $x_1$ and $x_3$ have the same status in the system (\ref{eqt:eqtsys_org}) and so do $x_2$ and $x_4$, it follows that $\mathcal{N}_{(-1,-1,1,1)}=\mathcal{N}_{(1,-1,-1,1)}=\mathcal{N}_{(-1,1,1,-1)}=\mathcal{N}_{(1,1,-1,-1)}$.
		\item When $(\chi(x_1),\chi(x_2),\chi(x_3),\chi(x_4))=(1,-1,1,-1)$, the system (\ref{eqt:eqtsys_org}) can be reduced to
		$$\left\{
		\begin{array}{l}
			x_1-x_2+x_3-x_4=0,\\
			x_1^2+x_3^2=0. \\
		\end{array}
		\right.
		$$
		Obviously, $x_1^2+x_3^2=0$ has no solution when $x_1\neq0$, $x_3\neq0$, which means $\mathcal{N}_{(1,-1,1,-1)}=0$. Besides, it is easy to check that $\mathcal{N}_{(-1,1,-1,1)}=0$.
		\end{enumerate}
	In short, we have $\mathcal{N}_{(1,1,-1,-1)}=\mathcal{N}_{(-1,-1,1,1)}=\mathcal{N}_{(1,-1,-1,1)}=\mathcal{N}_{(-1,1,1,-1)}=\frac{1}{4}(p^n-1)^2$, $\mathcal{N}_{(1,-1,1,-1)}=\mathcal{N}_{(-1,1,-1,1)}=0$.
	\item In this case, we consider there are exactly three nonsquare elements among $x_1,x_2,x_3$ and $x_4$. Suppose that $(\chi(x_1),\chi(x_2),\chi(x_3),\chi(x_4))=(1,-1,-1,-1)$, then the second equality in (\ref{eqt:eqtsys_org}) becomes $2x_1^2=0$. It follows that $x_1=0$, a contradiction to $x_i\neq0$ for $1\leqslant i\leqslant4$. Thus, $\mathcal{N}_{(1,-1,-1,-1)}=0$. By the same procedure, the desired result follows. Thus, we have $\mathcal{N}_{(1,-1,-1,-1)}=\mathcal{N}_{(-1,1,-1,-1)}=\mathcal{N}_{(-1,-1,1,-1)}=\mathcal{N}_{(-1,-1,-1,1)}=0$.
	\item When $(\chi(x_1),\chi(x_2),\chi(x_3),\chi(x_4))=(-1,-1,-1,-1)$, the system (\ref{eqt:eqtsys_org}) can be reduced to
		$$x_1-x_2+x_3-x_4=0.$$
		According to Lemma \ref{lem:eqtsys_y_xallmisus}, it follows that $\mathcal{N}_{(-1,-1,-1,-1)}=\frac{1}{16}\left((p^n-1)(p^{2n}-2p^n+5)\right)$.
\end{enumerate}
Above all, we have 
	\begin{align*}
	\mathcal{N}^{(0)}&=\sum\limits_{(i,j,k,l)\in\{\pm1\}^4}\mathcal{N}_{(i,j,k,l)}\\
	&=\mathcal{N}_{(1,1,1,1)}+4\mathcal{N}_{(1,1,1,-1)}+4\mathcal{N}_{(1,1,-1,-1)}+2\mathcal{N}_{(1,-1,1,-1)}+4\mathcal{N}_{(1,-1,-1,-1)}+\mathcal{N}_{(-1,-1,-1,-1)}\\
	&=\frac{(p^n-1)(p^n-2)}{2}+\frac{1}{2}(p^n-1)(p^n+1+(\chi(2)-1)\lambda_{p,n})+(p^n-1)^2+\frac{1}{16}\left((p^n-1)(p^{2n}-2p^n+5)\right)\\
	&=\frac{1}{16}\left((p^n-1)(p^{2n}+26p^n-23+4(\chi(2)-1)\lambda_{p,n})\right).
	\end{align*}
In the following, we consider the cases that there exists some $x_i=0$, where $i\in\{1,2,3,4\}$ to evaluate $\mathcal{N}^{(i)}$ with $1\leqslant i\leqslant4$ as follows.
\begin{enumerate}
	\item Obviously, $(0,0,0,0)$ is a solution of (\ref{eqt:eqtsys_org}).  Suppose that there are exactly three variables $x_i$ taking the value $0$ in a solution, it can be deduced that the solution must be $(0,0,0,0)$. This implies (\ref{eqt:eqtsys_org}) can not have a solution with exactly three variables being $0$ and the rest one being nonzero. Thus, we have $\mathcal{N}^{(4)}=1$ and $\mathcal{N}^{(3)}=0$.
	\item In this case, we consider the condition that there are exactly two variables $x_i$ taking the value $0$ in a solution. If $x_1=x_2=0$, then the system (\ref{eqt:eqtsys_org}) can be reduced to
	$$\left\{
	\begin{array}{l}
		x_3-x_4=0,\\
		(1+\chi(x_3))x_3^2-(1+\chi(x_4))x_4^2=0. \\
	\end{array}
	\right.
	$$
	Then we obtain the solutions in the form of $(0,0,x,x)$ with $x\in\gf_{p^n}^*$ and the number of such solutions is $p^n-1$. By the same token, the quadruples in the form of $(x,x,0,0)$, $(0,x,x,0)$ and $(x,0,0,x)$ are all solutions and the number of each is $p^n-1$. In short, $\mathcal{N}^{(2)}=4(p^n-1)$.
	\item In this case, we consider the condition that there is exactly one variable $x_i$ taking the value 0 in a solution. If $x_4=0$, then the system (\ref{eqt:eqtsys_org}) can be reduced to
	\begin{equation}\label{eqt:eqtsys_one0}
	\left\{
	\begin{array}{l}
		x_1-x_2+x_3=0,\\
		(1+\chi(x_1))x_1^2-(1+\chi(x_2))x_2^2+(1+\chi(x_3))x_3^2=0. \\
	\end{array}
	\right.
	\end{equation}
	For a solution $(x_1,x_2,x_3)\in(\gf_{p^n}^*)^3$ of (\ref{eqt:eqtsys_one0}), let $\mathtt{N}_{(i,j,k)}$ denote the number of solutions $(x_1,x_2,x_3)\in(\gf_{p^n}^*)^3$ of the system of equations (\ref{eqt:eqtsys_one0}) when $(\chi(x_1),\chi(x_2),\chi(x_3))=(i,j,k), i,j,k\in\{\pm1\}$. Next, we examine the solutions of this system (\ref{eqt:eqtsys_one0}) in the following eight cases.
	\begin{enumerate}
		\item When $(\chi(x_1),\chi(x_2),\chi(x_3))=(1,1,1)$, the system (\ref{eqt:eqtsys_one0}) can be reduced to
			$$\left\{
				\begin{array}{l}
					x_1-x_2+x_3=0,\\
					x_1^2-x_2^2+x_3^2=0. \\
				\end{array}
			\right.
			$$
			Let $y_i=\frac{x_i}{x_3}$, we shall calculate the number of solutions of
			$$\left\{
				\begin{array}{l}
					y_1-y_2+1=0,\\
					y_1^2-y_2^2+1=0, \\
				\end{array}
			\right.
			$$
			with $\chi(y_1)=\chi(y_2)=1$. Solving the equations above we get $y_1=0$, which contradicts to $\chi(y_1)=1$. Therefore, $\mathtt{N}_{(1,1,1)}=0$.
		\item Here we discuss that there is exactly one nonsquare element among $x_1,x_2,x_3$. 
			\begin{enumerate}
			\item When $(\chi(x_1),\chi(x_2),\chi(x_3))=(1,1,-1)$, (\ref{eqt:eqtsys_one0}) can be reduced to
			$$\left\{
				\begin{array}{l}
					x_1-x_2+x_3=0,\\
					x_1^2-x_2^2=0. \\
				\end{array}
			\right.
			$$
			Solving the system of equations above we get $x_3=0$, which is a contradiction. Therefore, $\mathtt{N}_{(1,1,-1)}=0$. Similarly, we have $\mathtt{N}_{(-1,1,1)}=0$.\\
			\item When $(\chi(x_1),\chi(x_2),\chi(x_3))=(1,-1,1)$,  the system (\ref{eqt:eqtsys_one0}) can be reduced to
			$$\left\{
				\begin{array}{l}
					x_1+x_3=0,\\
					x_1^2+x_3^2=0. \\
				\end{array}
			\right.
			$$
			Since $x_1^2+x_3^2=0$ has no solution in $(\gf_{p^n}^*)^2$, we have $\mathtt{N}_{(1,-1,1)}=0$.
			\end{enumerate}
		Thus, we have $\mathtt{N}_{(1,1,-1)}=\mathtt{N}_{(1,-1,1)}=\mathtt{N}_{(-1,1,1)}=0$.
		\item Here we discuss that there are exactly two nonsquare elements among $x_1,x_2,x_3$. When $(\chi(x_1),\chi(x_2),\chi(x_3))=(1,-1,-1)$, the second equation of (\ref{eqt:eqtsys_one0}) becomes $x_1^2=0$. It follows that $x_1=0$, a contradiction to $x_1\in\gf_{p^n}^*$. Thus, $\mathtt{N}_{(1,-1,-1)}=0$ and similarly, we can get $\mathtt{N}_{(-1,1,-1)}=\mathtt{N}_{(-1,-1,1)}=0$.
		\item When $(\chi(x_1),\chi(x_2),\chi(x_3))=(-1,-1,-1)$, (\ref{eqt:eqtsys_one0}) can be reduced to
			$$x_1-x_2+x_3=0.$$
		Set $y_i=\frac{x_i}{x_3}$, then $\chi(y_1)=\chi(y_2)=1$ and we shall consider the equation below for the first step
			$$y_1-y_2+1=0.$$
		Therefore, $(y_1,y_2)$ is a desired solution if and only if 
		$$\chi(y_1)=1~\text{and}~\chi(y_1+1)=1.$$
		And the number of such solutions is
				$$\frac{1}{4}\sum\limits_{y_1\in\gf_{p^n}^*,y_1\neq-1}(1+\chi(y_1))\cdot(1+\chi(y_1+1))=\frac{p^n-3}{4}.$$
		Therefore, $\mathtt{N}_{(-1,-1,-1)}=\frac{p^n-1}{2}\cdot\frac{p^n-3}{4}=\frac{1}{8}(p^n-3)(p^n-1)$.
	\end{enumerate}
	Based on the analysis above, when $x_4=0$, (\ref{eqt:eqtsys_org}) has $\frac{1}{8}(p^n-3)(p^n-1)$ solutions containing exactly one zero. As $x_2$ and $x_4$ share the same status in (\ref{eqt:eqtsys_org}), when $x_2=0$, (\ref{eqt:eqtsys_org}) has $\frac{1}{8}(p^n-3)(p^n-1)$ such solutions. When $x_3=0$, the system (\ref{eqt:eqtsys_org}) can be reduced to
		$$\left\{
				\begin{array}{l}
					x_1-x_2-x_4=0, \\
					(1+\chi(x_1))x_1^2-(1+\chi(x_2))x_2^2-(1+\chi(x_4))x_4^2=0.
				\end{array}
			\right.
		$$
		This system is equivalent to
		$$\left\{
				\begin{array}{l}
					x_2-x_1+x_4=0, \\
					(1+\chi(x_2))x_2^2-(1+\chi(x_1))x_1^2+(1+\chi(x_4))x_4^2=0.
				\end{array}
			\right.
		$$
		By a simple comparison, the number of solutions of (\ref{eqt:eqtsys_org}) when $x_1=0$ equals that of (\ref{eqt:eqtsys_org}) when $x_4=0$. Given the equivalent role of $x_1$ and $x_3$, the same conclusion holds when any solution features solely $x_3=0$. In short, we have $\mathcal{N}^{(1)}=\frac{1}{2}(p^n-3)(p^n-1)$.
\end{enumerate}
In conclusion, we have 
\begin{align*}
N_4&=\sum\limits_{i=0}^4\mathcal{N}^{(i)} \\
&=\frac{1}{16}\left((p^n-1)\left(p^{2n}+26p^n-23+4(\chi(2)-1)\lambda_{p,n}\right)\right)+\frac{1}{2}\left((p^n-3)(p^n-1)\right)+4(p^n-1)+0+1 \\
&=\frac{1}{16}\left((p^n-1)\left(p^{2n}+34p^n+17+4(\chi(2)-1)\lambda_{p,n}\right)\right)+1.
\end{align*}
\end{proof}

\subsection{The differential spectrum of $f_1$}
Based on the analysis in subsections \ref{sec:subsec1} and \ref{sec:subsec2}, we present the differential spectrum of $f_1$ as follows.
\begin{theorem}\label{thm:DifSpm}
	Let $p^n\equiv3~(mod~4)$. The differential spectrum of $f_1(x)=x^{2}+x^{\frac{p^n+3}{2}}$ over $\gf_{p^n}$ is
	\begin{align*}
		\mathbb{S}_{f_1}=\Big[
			&\omega_0=\frac{1}{8}\left((p^n-1)\left(3p^n+3+(\chi(2)-1)\lambda_{p,n}\right)\right), \\
			&\omega_1=\frac{1}{4}\left((p^n-1)\left(2p^n-2-(\chi(2)-1)\lambda_{p,n}\right)\right), \\
			&\omega_2=\frac{1}{8}\left((p^n-1)\left(p^n+1+(\chi(2)-1)\lambda_{p,n}\right)\right),\\
			&\omega_{\frac{p^n+1}{4}}=(p^n-1),\\
			&\omega_{p^n}=1
		\Big].
	\end{align*}
\end{theorem}
\begin{proof}
	By Theorem \ref{thm:sol_eqts}, we have the following system of equations pertaining to $f_1$
	\begin{equation}\label{subsec_difspm:eqtsys_f1}
	\left\{
		\begin{array}{l}
			\sum\limits_{i=0}^{q}\omega_i=q^2, \\
			\sum\limits_{i=0}^{q}i\omega_i=q^2, \\
			\sum\limits_{i=0}^{q}i^2\omega_i=N_4.
		\end{array}
	\right.
	\end{equation}
It is obvious that $\omega_q=1$. By Theorem \ref{uniformity}, we have $\omega_{\frac{p^n+1}{4}}=p^n-1$,  and $\omega_i=0$ for $3\leqslant i\leqslant q-1$, $i\neq \frac{p^n+1}{4}$. The value of $N_4$ was determined in Theorem \ref{thm:N4_f1}. By substituting certain values, the system (\ref{subsec_difspm:eqtsys_f1}) can be rewritten as follows
	$$\left\{
		\begin{array}{l}
			\omega_0+\omega_1+\omega_2=q^2-\omega_{\frac{p^n+1}{4}}-\omega_q, \\
			\omega_1+2\omega_2=q^2-\left(\frac{p^n+1}{4}\right)\omega_{\frac{p^n+1}{4}}-q\omega_q, \\
			\omega_1+2^2\omega_2=N_4-\left(\frac{p^n+1}{4}\right)^2\omega_{\frac{p^n+1}{4}}q^2\omega_q.
		\end{array}
	\right.
	$$
	The desired result follows by solving the above system. We finish the proof.
\end{proof}

\begin{remark}\label{rmk:spmwithchi2}
	Note that $n$ is odd when $p^n\equiv3~(mod~4)$. Thus, when $p\equiv7~(mod~8)$, we have $\chi(2)=1$, then the differential spectrum of $f_1$ can be expressed without $\lambda_{p,n}$, which is
	\begin{align*}
		\bigg[
			\omega_0=\frac{3}{8}\left(p^{2n}-1\right),~
			\omega_1=\frac{1}{2}(p^n-1)^2,~
			\omega_2=\frac{1}{8}\left(p^{2n}-1\right),~
			\omega_{\frac{p^n+1}{4}}=p^n-1,~ 
			\omega_{p^n}=1
		\bigg].
	\end{align*}
	When $p\equiv3~(mod~8)$, we have $\chi(2)=-1$, then the differential spectrum of $f_1$ is
	\begin{align*}
		\bigg[
			&\omega_0=\frac{1}{8}\left((p^n-1)\left(3p^n+3-2\lambda_{p,n}\right)\right),~
			\omega_1=\frac{1}{2}\left((p^n-1)\left(p^n-1+\lambda_{p,n}\right)\right), \\
			&\omega_2=\frac{1}{8}\left((p^n-1)\left(p^n+1-2\lambda_{p,n}\right)\right),~
			\omega_{\frac{p^n+1}{4}}=p^n-1,~
			\omega_{p^n}=1
		\bigg].
	\end{align*}
\end{remark}

\begin{remark}
	According to Lemma \ref{lem:WeilB}, $\left|\lambda_{p,n}\right|\leqslant2q^{\frac{1}{2}}$. With Remark \ref{rmk:spmwithchi2} we have $\omega_2\geqslant\frac{1}{8}\left[(p^n-1)(p^n+1-4p^{\frac{n}{2}})\right]$. This implies $\omega_2>0$ if and only if $p^n\geqslant11$. Consequently, 
	\begin{enumerate}
	\item when $p^n=3$, the differential uniformity of $f_1$ is $1$ and based on Table \ref{tab:valuesoflambda}, the differential spectrum of $f_1$ is $[\omega_0=2,~\omega_1=6,~\omega_3=1]$, which represents the unique condition of a PN function within class of binomials described in Theorem \ref{thm:DifSpm};
	\item when $p^n=7$, the differential uniformity of $f_1$ is $2$ and based on Remark \ref{rmk:spmwithchi2}, the differential spectrum of $f_1$ is $[\omega_0=18,~\omega_1=12,~\omega_2=18,\omega_7=1]$, which represents the unique condition of an APN function within the class of binomials described in Theorem \ref{thm:DifSpm};
	\item when $p^n\geqslant11$, the function in Theorem \ref{thm:DifSpm} is exactly a locally APN function.
	\end{enumerate}
\end{remark}

In what follows, we give some examples to verify our results.
\begin{example} Let $p=3$, $n=5$. Then $p^n-1=242$, $\chi(2)=-1$, $\lambda_{p,n}=2$.
	By Theorem \ref{thm:DifSpm}, the differential spectrum of $f_1$ is
   $$\left[\omega_0=22022,~\omega_1=29524,~\omega_2=7260,~\omega_{61}=242,~\omega_{243}=1\right],$$
   which coincides with the result calculated directly by MAGMA.
\end{example}

\begin{example} Let $p=7$, $n=3$. Then $p^n-1=342$, $\chi(2)=1$, $\lambda_{p,n}=20$.
	By Theorem \ref{thm:DifSpm}, the differential spectrum of $f_1$ is
   $$\left[\omega_0=44118,~\omega_1=58482,~\omega_2=14706,~\omega_{86}=342,~\omega_{343}=1\right],$$
   which coincides with the result calculated directly by MAGMA.
\end{example}

\begin{example} Let $p=11$, $n=3$. Then $p^n-1=1330$, $\chi(2)=-1$, $\lambda_{p,n}=58$.
	By Theorem \ref{thm:DifSpm}, the differential spectrum of $f_1$ is
   $$\left[\omega_0=645050,~\omega_1=923020,~\omega_2=202160,~\omega_{333}=1330,~\omega_{1331}=1\right],$$
   which coincides with the result calculated directly by MAGMA.
\end{example}

\section{Concluding remarks}\label{sec:con}
In this paper, we conducted an in-depth investigation of the differential properties of the function $f_u(x)=x^{\frac{p^n+3}{2}}+ux^2$ for $u=\pm1$. We expressed the differential spectrum of $f_{\pm1}$ in terms of quadratic character sums. This complemented the work on the differential properties of the family of the binomial in \cite{Budaghyan2024ArithmetizationorientedAP}. In the process of calculating the aimed spectrum, we solved several systems of equations that could be of use in future research or other contexts. Additionally, we extend the properties of the differential spectrum property of a power function to that of any cryptographic function, and it may be used in calculating the differential spectrum of any other polynomial.

\bibliographystyle{IEEEtranS}

\bibliography{A+note+on+the+differential+spectrum+of+a+class+of+locally+APN+function}
\end{document}